\documentclass{article}
\usepackage[T1]{fontenc}
\usepackage{tgpagella}
\usepackage{graphicx,amsmath,amsfonts,amsthm,amscd,amssymb,bm,url,color,latexsym}
\usepackage{mathrsfs}

\renewenvironment{proof}[1][\proofname]{\noindent {\bfseries #1.}  }{\qed}
\usepackage{fullpage}
\usepackage[small,bf]{caption}
\usepackage{subcaption}
\usepackage{bbm}
\usepackage{microtype}
\usepackage{algorithm}
\usepackage{algorithmicx}
\usepackage{algpseudocode}
\usepackage{verbatim}
\usepackage{framed}
\usepackage{comment}
\usepackage{tikz}
\usepackage{fge}
\usepackage{mathtools}
\usepackage{enumerate}
\usepackage{multirow}

\allowdisplaybreaks

\usepackage{hyperref}
\hypersetup{
    colorlinks=true,%
    citecolor=blue,%
    filecolor=blue,%
    linkcolor=blue,%
    urlcolor=blue
}
\usepackage[toc, page]{appendix}

\usepackage[capitalize,nameinlink]{cleveref}

\newtheorem{theorem}{Theorem}[section]
\newtheorem{lemma}[theorem]{Lemma}

\newtheorem{proposition}[theorem]{Proposition}
\newtheorem{definition}[theorem]{Definition}

\renewcommand{\mathbf}{\boldsymbol}

\newcommand{\mb}{\mathbf}
\newcommand{\mc}{\mathcal}

\newcommand{\bb}{\mathbb}

\newcommand{\set}[1]{\left\{ #1 \right\}}

\newcommand{\reals}{\bb R}

\newcommand{\eps}{\varepsilon}
\newcommand{\R}{\reals}

\newcommand{\Z}{\bb Z}
\newcommand{\N}{\bb N}

\newcommand{ \brac }[1]{\left[ #1 \right]}
\newcommand{ \Brac }[1]{\left\lbrace #1 \right\rbrace}
\newcommand{ \paren }[1]{ \left( #1 \right) }

\DeclareMathOperator{\diag}{diag}

\DeclareMathOperator{\vect}{vec}

\newcommand{\wh}{\widehat}
\newcommand{\wt}{\widetilde}
\newcommand{\ol}{\overline}

\newcommand{\norm}[2]{\left\| #1 \right\|_{#2}}
\newcommand{\abs}[1]{\left| #1 \right|}
\newcommand{\innerprod}[2]{\left\langle #1,  #2 \right\rangle}

\numberwithin{equation}{section}

\title{Dual-Reference Design for \\
Holographic Coherent Diffraction Imaging\thanks{A preliminary version is published in International Conference on Sampling and Applications, 2019~\cite{HCDI_SampTA}. }}

\author{David A. Barmherzig\thanks{Institute for Computational and Mathematical Engineering, Stanford University, Stanford, CA 94305, U.S.A.}
        \and Ju Sun\thanks{Department of Mathematics, Stanford University, Stanford, CA 94305, U.S.A.}
        \and Po-Nan Li\thanks{Department of Electrical Engineering, Stanford University, Stanford, CA 94305, U.S.A.}
                \and T.J. Lane\thanks{SLAC National Accelerator Laboratory, Menlo Park, CA 94025, U.S.A.}
        \and Emmanuel J. Cand\`{e}s\thanks{Department of Mathematics and Department of Statistics, Stanford University, Stanford, CA 94305, U.S.A.}
}
\date{}

\date{  \quad Revised: \today}
\date{\today}

\begin{document}
\maketitle

\begin{abstract}
A new reference design is introduced for holographic coherent diffraction imaging. This consists in two references\textemdash ``block'' and ``pinhole'' shaped regions\textemdash placed adjacent to the imaging specimen. An efficient recovery algorithm is provided for the resulting holographic phase retrieval problem, which is based on solving a structured, overdetermined linear system. Analysis of the expected recovery error on noisy data, which is contaminated by Poisson shot noise, shows that this simple modification synergizes the individual references and hence leads to uniformly superior performance over single-reference schemes. Numerical experiments on simulated data confirm the theoretical prediction, and the proposed dual-reference scheme achieves a smaller recovery error than leading single-reference schemes.
\end{abstract}

\section{Introduction}

\subsection{Holographic CDI and holographic phase retrieval}

Coherent Diffraction Imaging (CDI) is a scientific imaging technique used for resolving nanoscale scientific specimens, such as macroviruses, proteins, and crystals~\cite{CDI-orig}. In CDI, a coherent radiation source (often an X-ray beam) is incident on a specimen and gets diffracted. The resulting photon flux is then measured at a far-field detector, and the measured data are approximately proportional to the squared magnitudes of the Fourier transform of the wave field within the diffraction area. One can then determine the specimen's structure by solving the \textit{phase retrieval} problem, the mathematical inverse problem of recovering a signal from its squared Fourier magnitudes.
The physical phase retrieval problem can be stated symbolically as
\begin{align}  \label{eq:pr_symbol}
\begin{split}
&\textbf{Given} \quad \big{|}\wh{X}(\omega)\big{|}^2 \doteq \abs{\int_{t \in T} X(t)e^{-i \innerprod{\omega}{t}}}^2\quad \text{for}\; \omega \in \Omega \\
&\textbf{Recover} \quad X
\end{split},
\end{align}
where $T$ and $\Omega$ are the (possibly multidimensional) domains of the specimen and its Fourier transform, respectively. Since the detectors used in practical CDI have only a finite number of pixels, and algorithmic phase retrieval is often performed on digital computers, phase retrieval is frequently stated in the discrete form: in \cref{eq:pr_symbol}, $\Omega$ and $T$ are finite-size arrays and the Fourier transform is implemented as the discrete Fourier transform. The discrete formulation is a reasonable proxy for the continuous formulation, as explained in \cref{sec:c2d_app}. We adopt this direct discrete formulation in this paper.

In a variant of CDI known as Holographic CDI, a ``reference'' portion of the diffraction area is a priori known from experimental design (e.g. see \cref{FH-CDI,single-ref}). Typically, the reference portion is simply a geometric shape cut out from the apparatus surrounding the specimen. The resulting inverse problem, in which a portion of the signal to be recovered is already known, is the \emph{holographic phase retrieval} problem.

\begin{figure}[!htbp] \label{FH-CDI}
    \centering
        \includegraphics[width=0.5\textwidth]{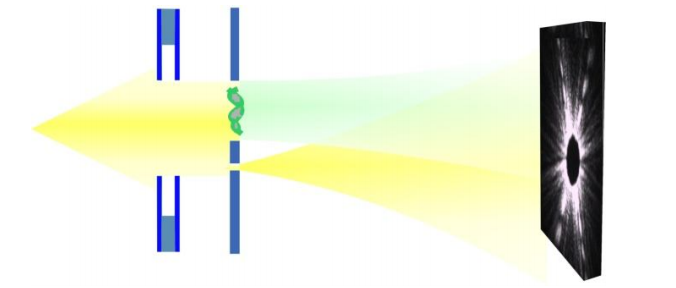}
        \caption{Holographic CDI schematic. The upper portion of the diffraction area contains the imaging specimen of interest, and the lower portion consists of a known ``reference'' shape. Image courtesy of~\cite{FT-Cambridge}.}
    \label{FH-CDI}
\end{figure}

For any reference choice satisfying mild assumptions, solving the holographic phase retrieval problem amounts to solving a structured linear system~\cite{HologPROptREF}. However, different reference choices will lead to different noise stability performances. Particularly, our previous work~\cite{HologPROptREF} revealed, both theoretically and empirically, the relative merits of two popular references: the block reference $R_B$ (see \cref{block-ref}) performs favorably on data with low-frequency dominant spectra, whereas the pinhole reference (see \cref{pinhole-ref}) has an edge for data with flat spectra. It is a natural question if the respectively advantages can be combined, and how.


\begin{figure}[!htbp] \label{single-ref}
    \centering
        \includegraphics[width=0.4\textwidth]{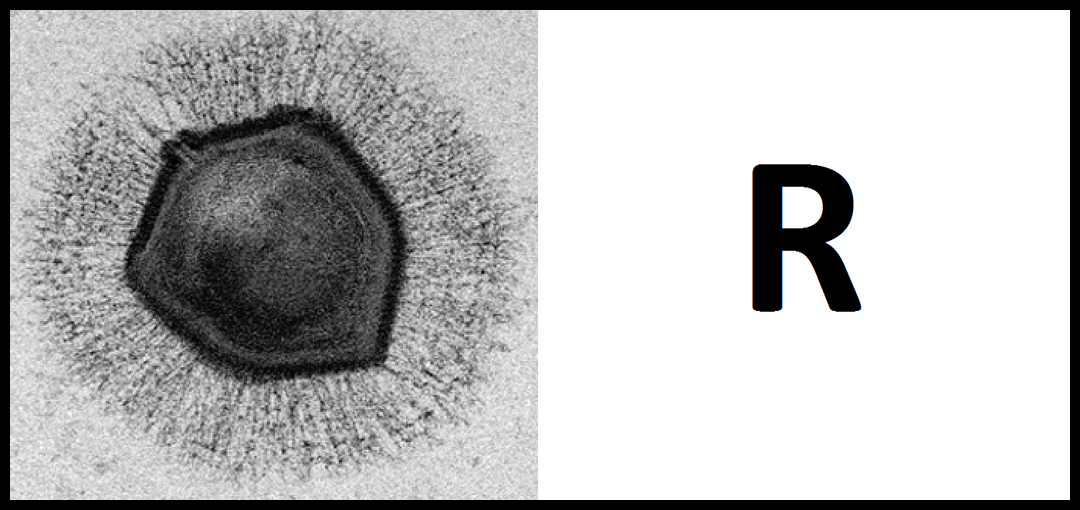}
        \caption{Schematic of the diffraction area in Holographic CDI containing a specimen and a known (single) reference portion. The specimen shown is the Mimivirus, courtesy of \cite{Mimivirus}.}
    \label{single-ref}
\end{figure}

\begin{figure}[!htbp]
\centering
\begin{subfigure}{0.2\textwidth}
        \includegraphics[width=\textwidth]{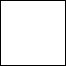}
        \caption{}
          \label{block-ref}
    \end{subfigure}
        \begin{subfigure}{0.2\textwidth}
        \includegraphics[width=\textwidth]{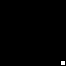}
        \caption{}
        \label{pinhole-ref}
    \end{subfigure}
    \caption{Two popular choices for the reference R shown in \cref{single-ref} are the block reference (\cref{block-ref}) and the pinhole reference (\cref{pinhole-ref}).}
\end{figure}

\subsection{Our contributions}
In this paper, we answer the question in the affirmative and show that a simple augmentation of the block and pinhole references actually works as desired. A recovery algorithm is adapted from the referenced deconvolution algorithm introduced in~\cite{HologPROptREF}.  Both theoretical (\cref{sec:err}) and empirical (\cref{sec:exp}) results confirm the effectiveness of the proposed augmentation scheme. From hindsight, this is still a bit surprising, as the recovery error depends on the reference choice in a complicated manner; see \cref{eqn:linear-exp-fund}.

\subsection{A word on notation}
We mostly use standard mathematical notations, with several special ones highlighted as below. Matlab notations $[A, B]$ and $[A ; B]$ are used to mean horizontal and  vertical concatenation of matrices, respectively. Similarly, notations such as $A\paren{k, :}$ and $A\paren{:, k}$ are used to index rows and columns of matrices, respectively. We use $\otimes$ to denote the matrix Kronecker product.

\section{Dual-reference design and recovery algorithm}

\subsection{Setup and algorithm}
\begin{definition}
The block reference $R_B \in \R^{n \times n}$ and the pinhole reference $R_P \in \R^{n \times n}$ are defined respectively by
\begin{equation} \label{eqn:empty-space-ref}
R_B(t_1,t_2)=1, \quad t_1, t_2 \in \{0,\dots,n-1\},
\end{equation}
and
\begin{equation} \label{eqn:Four-holog-ref}
R_P(t_1,t_2) =
        \begin{cases}
            1, & \quad t_1=t_2 =n-1 \\
            0, & \quad \text{otherwise}
        \end{cases}.
\end{equation}
\end{definition}
The two references are visualized in \cref{block-ref} and \cref{pinhole-ref}, respectively. Suppose that $X \in \mathbb{C}^{n \times n}$ is an ``unknown'' specimen. Consider $\mathcal{X} \in \mathbb{C}^{2n \times 2n}$ given by:
\begin{equation} \label{eqn:dual-ref-setup}
\mathcal{X}
=
\begin{bmatrix}
    X & R_B \\
    R_P & \mathbf{0}_{n \times n}
\end{bmatrix},
\end{equation}
where $\mb 0_{n \times n}$ is the $n \times n$ all-zero matrix.
We assume that the magnitudes of the entries of $X$ are within the interval $[0,1]$. By this convention, 0 values represent areas where the incoming beam is entirely blocked, and 1 values represent areas where the incoming beam passes through unimpeded---which would be ``empty space''.

\begin{figure}[!htbp] \label{dual-ref}
    \centering
        \includegraphics[width=0.4\textwidth]{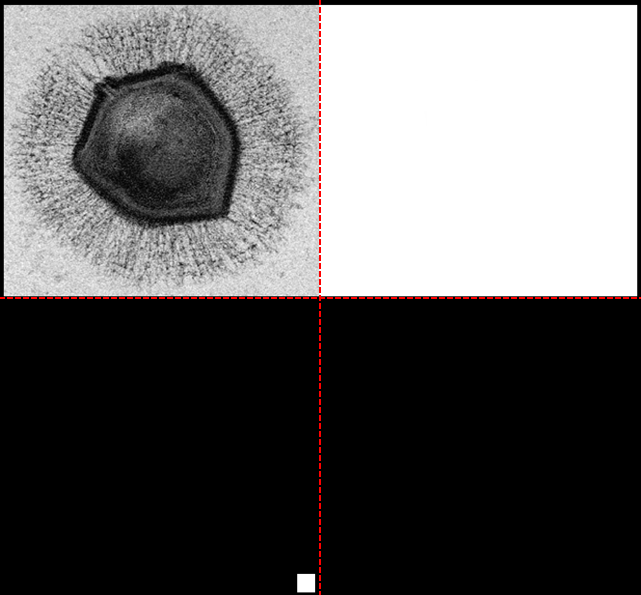}
        \caption{Schematic of the dual-reference. The red dotted line (added for illustration purposes) separates the four quadrants of the setup, as described by ~\cref{eqn:dual-ref-setup}. The specimen shown is the mimivirus, courtesy of \cite{Mimivirus}. }
    \label{dual-ref}
\end{figure}

Suppose that $m \geq 4n-1$ and that $Y=|\widehat{\mc X}|^2 \in \mathbb{C}^{m \times m}$ are the magnitudes of the $m \times m$ oversampled Fourier transform of $\mathcal{X}$\footnote{Here, the absolute value notation on $Y$ is understood in the pointwise sense. Also, we adopt the Matlab convention and assume the zero-frequency component is on the top-left corner of the data matrix. }. We seek to recover $X$ from $\wt Y$, which is a possibly noise-corrupted version of $Y$. We propose a recovery algorithm based on solving a structured linear system, which is effectively the \emph{referenced deconvolution} algorithm introduced in \cite{HologPROptREF} adapted to our current reference scheme.

\begin{enumerate}
\item Given $\wt Y$, apply an inverse Fourier transform ($\bb C^{m \times m} \mapsto \bb C^{(4n-1) \times  (4n-1)}$) to obtain $\wt A_{\mathcal{X}}$, the noisy autocorrelation of $\mathcal{X}$.\footnote{It is well-known that the inverse Fourier transform of the Fourier transform (with sufficient oversampling) magnitude squares of a signal is equal to the signal's autocorrelation~\cite{Oppenheim}.} This can be expressed as $\wt A_{\mathcal{X}}=\frac{1}{m^2}F^*\wt Y (F^*)^T$, where $F \in \mathbb{C}^{m\times (4n-1)}$ is given by $F(k,t)=e^{-2 \pi i k t/m}\; \forall\; (k, t) \in \{0,\dots, m-1\} \times \{-(2n-1), \dots, 2n-1\}$.

\item Let $\mc P_1 = [\mathbf{0}_{n\times n},I_n, \mathbf{0}_{n\times (2n-1)}]$ and $\mc P_2 = [I_n, \mathbf{0}_{n\times (3n-1)}]$. It follows that absent noise, $\mc P_1 \wt A_{\mc X} \mc P_2^\top \in \R^{n \times n}$ (resp., $\mc P_2 \wt A_{\mc X} \mc P_1^\top \in \R^{n \times n}$) is equal to the top-left quadrant of the cross-correlation of $X$ and $R_B$ (resp., $X$ and $R_P$). We thus denote this as $\wt C^{\diamond}_{[X,R_B]}$ (resp., $\wt C^{\diamond}_{[X,R_P]}$).

\item Let $M_{R_B}$ (resp., $M_{R_P}$) $\in \mathbb{R}^{n^2 \times n^2}$ be the matrix satisfying $\vect(C^{\diamond}_{[X,R_B]}) = M_{R_B}\vect(X)$ (resp., $\vect(C^{\diamond}_{[X,R_P]}) = M_{R_P}\vect(X)$).\footnote{The cross-correlations are linear in $X$, and hence such $M_{R_B}$ and $M_{R_P}$ exist for the noiseless cross-correlations $C^\diamond_{[X, R_B]}$ and $C^\diamond_{[X, R_P]}$. } It follows that $M_{R_B}= \mb 1_L \otimes \mb 1_L$, where $\mb 1_L \in \R^{n \times n}$ is the lower-triangular matrix consisting of all ones on and below the main diagonal, and that $M_{R_P}=I_{n^2}$. Let $M = [M_{R_B};  M_{R_P}]$ and $b = [
\vect(\wt C^{\diamond}_{[X,R_B]}) ; \vect(\wt C^{\diamond}_{[X,R_P]})]$. The signal $X$ is estimated as the solution to the least-squares problem
\begin{equation*}
\wt X = \mathop{\arg\min}_{X \in \bb C^{n \times n}} \; \norm{M\vect(X)-b}{}^2.
\end{equation*}
Analytically, this is given by
\begin{equation*}
\vect(\wt X)=M^\dagger b = (M^TM)^{-1}M^Tb.
\end{equation*}
\end{enumerate}
Combining these steps and the well-known identity that $A = BCD \Longleftrightarrow \vect(A)=(D^T \otimes B)\vect(C)$, we have
\begin{equation} \label{eqn:fund}
\vect(\wt X)=T_{R_{B,P}}\vect(\wt Y),
\end{equation}
where
\begin{equation} \label{eqn:T-shortform}
T_{R_{B,P}}=\frac{1}{m^2}M^\dagger  \begin{bmatrix} \mc P_2 F^* \otimes \mc P_1 F^* \\ \mc P_1 F^* \otimes \mc P_2 F^* \end{bmatrix}.
\end{equation}
Note that $\wt X$ is linear in $\wt Y$ and the proposed algorithm recovers $X$ exactly when there is no noise.

\subsection{Computation}
Recall that $M \in \R^{2n^2 \times n^2}$. Thus, taking the inverse Fourier transform and solving the above linear system would cost $O(n^6 + m^2 \log m)$. Below, we show that for our specific $M = [\mb 1_L \otimes \mb 1_L; I_{n^2}]$, the cost can be significantly lower when the structure in $M$ is exploited. We need the following result to proceed.

\begin{lemma}[Chapter 1 of~\cite{Strang_CSE}]  \label{lem:svd_1L}
Let $\mb 1_L \in \R^{n \times n}$ be the lower triangular matrix with ones on and below the main diagonal. The singular value decomposition $\mb 1_L=U\Sigma V^T$ is such that for any $s = 0, \dots, n-1$, $U$ and $V$ have columns given by
\begin{equation*}
U(t, s) = \frac{1}{\sqrt{\frac{n}{2} + \frac{1}{4}}}\sin\Brac{\frac{(s+\frac{1}{2})(t+1)}{n+\frac{1}{2}} \pi}  \quad \text{for}\; t = 0, \dots, n-1
\end{equation*}
\begin{equation*}
V(t, s) = \frac{1}{\sqrt{\frac{n}{2} + \frac{1}{4}}}\cos\Brac{\frac{(s+\frac{1}{2})(t+\frac{1}{2})}{n+\frac{1}{2}} \pi} \quad \text{for}\; t = 0, \dots, n-1
\end{equation*}
respectively, and $\Sigma$ has diagonal entries given by
\begin{equation*}
\sigma_s = \brac{2 - 2 \cos \paren{\frac{s+\frac{1}{2}}{n+\frac{1}{2}} \pi}}^{-1/2}.
\end{equation*}
\end{lemma}
Since $M = [\mb 1_L \otimes \mb 1_L; I_{n^2}]$ and writing the SVD of $\mb 1_L$ as $\mb 1_L = U \Sigma V^\top$ per \cref{lem:svd_1L}, we have
\begin{align*}
    M^\dagger
    & = \brac{\paren{\mb 1_L^\top \otimes \mb 1_L^\top} \paren{\mb 1_L \otimes \mb 1_L} + I_{n^2}}^{-1}  \brac{\mb 1_L^\top \otimes \mb 1_L^\top,  I_{n^2}} \\
    & = \brac{\paren{V\Sigma U^\top \otimes V\Sigma U^\top} \paren{U\Sigma V^\top \otimes U\Sigma V^\top} + I_{n^2}}^{-1} \brac{V\Sigma U^\top \otimes V\Sigma U^\top,  I_{n^2}}\\
    & = \brac{\paren{V\Sigma^2 V^\top} \otimes \paren{V \Sigma^2 V^\top} + I_{n^2}}^{-1} \brac{V\Sigma U^\top \otimes V\Sigma U^\top,  I_{n^2}} \quad (\text{mixed product property\footnotemark})\\
    & = \brac{\paren{V \otimes V} \paren{\Sigma^2 \otimes \Sigma^2} \paren{V^\top \otimes V^\top} + I_{n^2}}^{-1} \brac{\paren{V \otimes V} \paren{\Sigma \otimes \Sigma} \paren{U^\top \otimes U^\top}, I_{n^2}}.
\end{align*}
\footnotetext{For any $A, B$ and $C, D$ of compatible dimensions, $\paren{A \otimes B} \paren{C \otimes D} = \paren{AC} \otimes \paren{BD}$. }
Now that $V \otimes V$ is an orthogonal matrix, it follows
\begin{align*}
\brac{\paren{V \otimes V} \paren{\Sigma^2 \otimes \Sigma^2} \paren{V^\top \otimes V^\top} + I_{n^2}}^{-1}
= \paren{V \otimes V} \paren{\Sigma^2 \otimes \Sigma^2 + I_{n^2}}^{-1} \paren{V^\top \otimes V^\top}.
\end{align*}
Thus,
\begin{align}
   M^\dagger = \brac{\paren{V \otimes V} \paren{\Sigma^2 \otimes \Sigma^2 + I_{n^2}}^{-1} \paren{\Sigma \otimes \Sigma} \paren{U^\top \otimes U^\top}, \paren{V \otimes V} \paren{\Sigma^2 \otimes \Sigma^2 + I_{n^2}}^{-1} \paren{V^\top \otimes V^\top}}.
\end{align}
By \cref{eqn:T-shortform},
\begin{align} \label{eq:TR_key1}
T_{R_{B, P}} &  = \frac{1}{m^2} \paren{V \otimes V} \left[ \paren{\Sigma^2 \otimes \Sigma^2 + I_{n^2}}^{-1} \paren{\Sigma \otimes \Sigma}  \paren{U^\top \otimes U^\top}  \paren{\mc P_2 F^* \otimes \mc P_1 F^*}\right. \nonumber \\
& \qquad \left. +  \paren{\Sigma^2 \otimes \Sigma^2 + I_{n^2}}^{-1} \paren{V^\top \otimes V^\top} \paren{\mc P_1 F^* \otimes \mc P_2 F^*} \right] \nonumber \\
& = \frac{1}{m^2} \paren{V \otimes V}
    \left[\paren{\Sigma^2 \otimes \Sigma^2 + I_{n^2}}^{-1} \paren{\Sigma \otimes \Sigma} \paren{U^\top \mc P_2 F^* \otimes U^\top \mc P_1 F^*} \right. \nonumber\\
& \qquad \left. +  \paren{\Sigma^2 \otimes \Sigma^2 + I_{n^2}}^{-1} \paren{V^\top \mc P_1 F^* \otimes V^\top \mc P_2 F^*} \right].
\end{align}
For the final computation, we make repeated use of the property $A = BCD \Longleftrightarrow \vect\paren{A} = \paren{D^\top \otimes B} \vect\paren{C}$ to obtain that
\begin{align*}
	\vect\paren{\wt X}
	& = T_{R_{B, P}} \vect\paren{\wt Y} \nonumber \\
	& = \frac{1}{m^2} \paren{V \otimes V} \underbrace{\paren{\Sigma^2 \otimes \Sigma^2 + I_{n^2}}^{-1}  \brac{\vect\paren{\Sigma U^\top \mc P_1 F^* \wt{Y} \ol{F} \mc P_2^\top U \Sigma} + \vect\paren{V^\top \mc P_2 F^* \wt Y \ol{F} \mc P_1^\top V}}}_{\doteq \vect\paren{Q}\; \text{for a}\; Q \in \R^{n^2 \times n^2}}\\
    & = \frac{1}{m^2} \vect\paren{V Q V^\top},
\end{align*}
which obviously can be computed in $O(n^3 + m^2 \log m)$ time, as against $O(n^6 + m^2 \log m)$.

\section{Analysis of the recovery error} \label{sec:err}
For any data $\wt Y$ following a known probability distribution, it follows from \cref{eqn:fund} that
\begin{align} \label{eqn:linear-exp}
\mathbb{E}\|\wt X - X\|_F^2
=\innerprod{T_R^* T_R}{\bb E \brac{\vect\paren{\wt Y}-\vect\paren{Y} } \brac{\vect\paren{\wt Y} -\vect\paren{Y} }^*},
\end{align}
where $\|\cdot\|_F^2$ denotes the Frobenius norm, and $\innerprod{\cdot}{\cdot}$ denotes the Frobenius inner product.

In CDI, measurements of photon flux at the detector are subject to quantum shot noise. This is due to intrinsic quantum fluctuations which cannot be removed in any measurement system. The resulting measurements follow the \text{Poisson shot noise} distribution given by
\begin{equation} \label{eqn:data}
 \wt Y \sim_{\mathrm{ind}} \frac{\norm{Y}{1}}{N_p}\mathrm{Pois}\Big{(}\frac{N_p}{\norm{Y}{1}}Y\Big{)},
\end{equation}
where $N_p$ is the expected (or nominal) number of photons reaching the detector, and $\norm{Y}{1}$ is understood as the $\ell_1$ norm of a vectorized version of $Y$ ~\cite{CDI-stats}.
As derived in \cite{HologPROptREF}, under this model
\begin{equation} \label{eqn:linear-exp-fund}
\mathbb{E}\|\wt X - X\|_F^2
= \frac{\norm{Y}{1}}{N_p} \innerprod{S_{R_{B,P}}}{Y},
\end{equation}
where
\begin{equation} \label{eq:scale-fact}
S_{R_{B,P}}=\mathrm{reshape}\big{(}\diag(T_R^*T_R),m,m\big{)},
\end{equation}
and $\mathrm{reshape}(\cdot,m,m)$ is the columnwise vector-to-matrix reshaping operator.

In the recovery error~\cref{eqn:linear-exp-fund}, both $Y$ and $S_{R_{B, P}}$ depend on the references in use. Empirically, for low-frequency dominant $X$ (e.g., CDI specimens shown in \cref{single-ref} and \cref{test-images}), the spectrum of $Y$ is similar to that of $X$, up to small variations in the magnitudes; see \cref{data-compar}.
\begin{figure}[!htbp] \label{data-compar}
    \centering
        \includegraphics[width=0.8\textwidth]{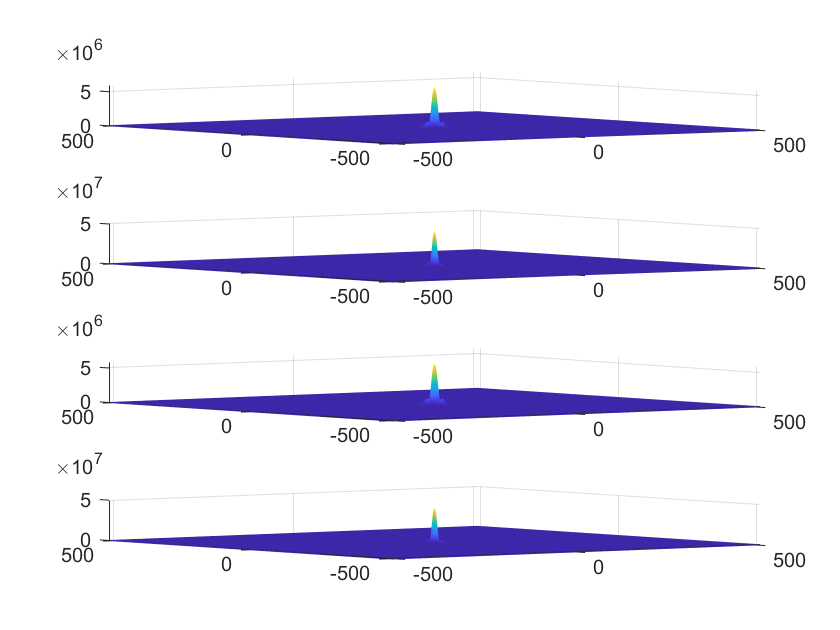}
        \caption{Top to bottom: squared magnitudes of the Fourier transform of the mimivirus~\cite{Mimivirus} itself, and that when the mimivirus is augmented with the block, pinhole, and dual-references, respectively (with $n=64$, and $m=1024$). These four spectra exhibit similar low-frequency dominance, and have entries of similar orders of magnitude.}
    \label{data-compar}
\end{figure}
This stability property of spectrum can be formally established for the single-reference setup $[X, R]$ (by expanding $|\wh{[X, R]}|^2$~\cite{HologPROptREF}), and likewise for our dual-reference setup. In contrast, the weighting factors in $S_R$ can vary by several orders of magnitude for different reference schemes, as shown in~\cref{ref-plot}.
\begin{figure}[!htbp] \label{ref-plot}
    \centering
    \begin{subfigure}[b]{0.3\textwidth}
        \includegraphics[width=\textwidth]{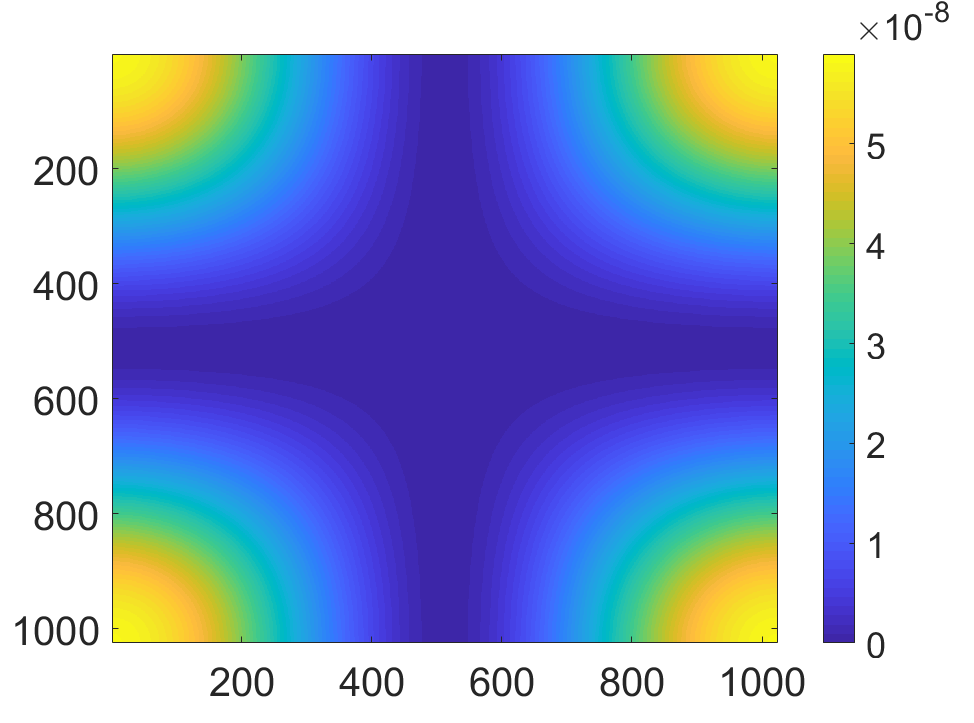}
        \caption{}
          \label{b-map}
    \end{subfigure}
    \begin{subfigure}[b]{0.3\textwidth}
        \includegraphics[width=\textwidth]{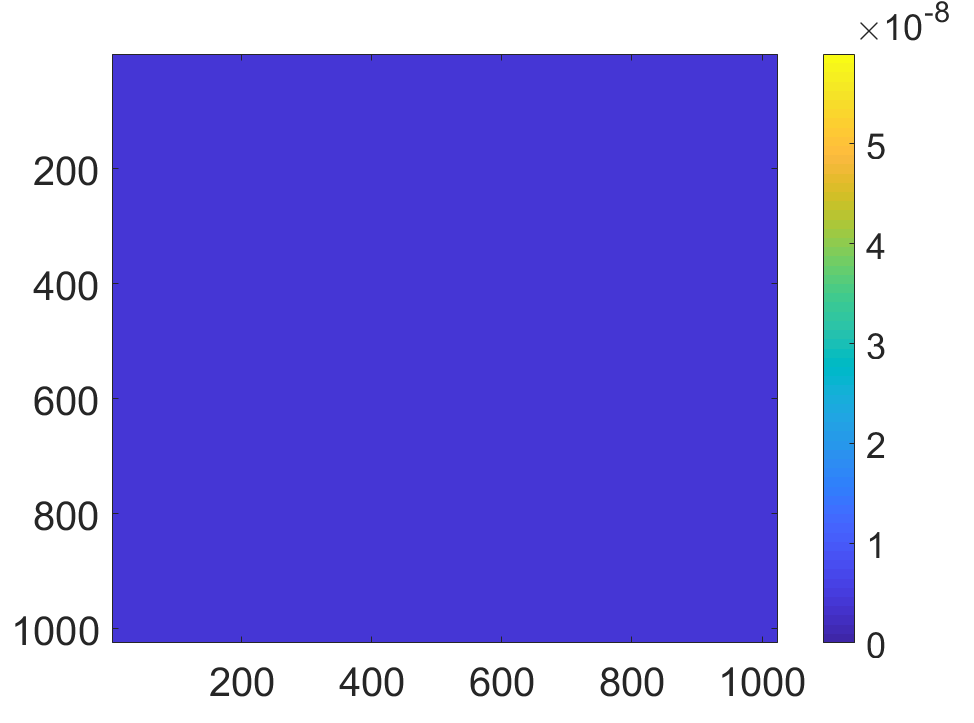}
        \caption{}
          \label{p-map}
    \end{subfigure}
    \begin{subfigure}[b]{0.3\textwidth}
        \includegraphics[width=\textwidth]{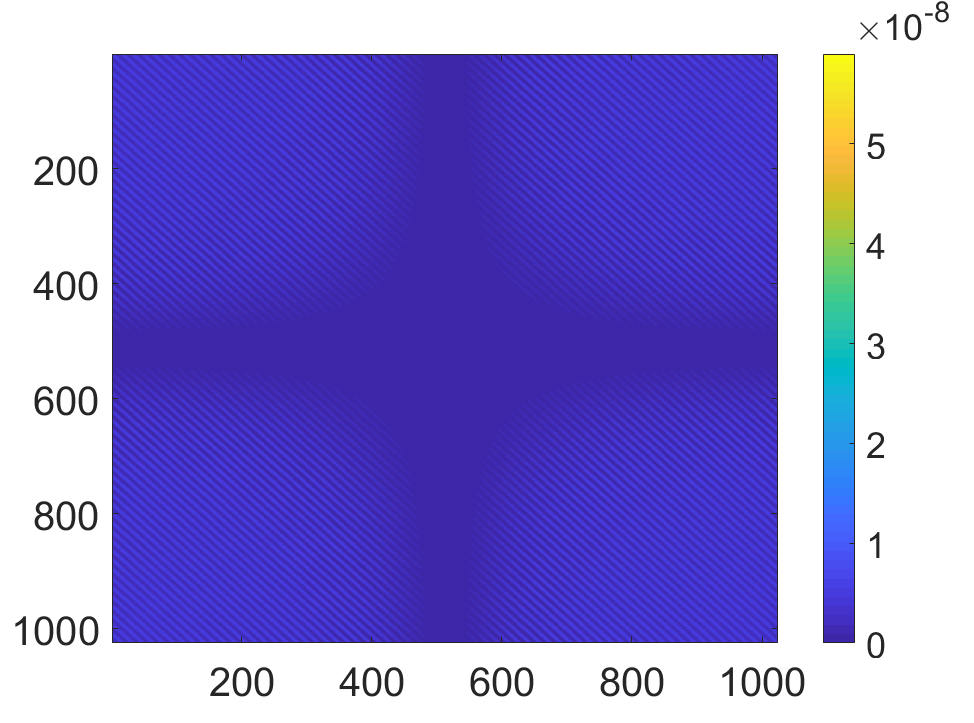}
        \caption{}
          \label{bp-map}
    \end{subfigure}
    \begin{subfigure}[b]{0.5\textwidth}
        \includegraphics[width=\textwidth]{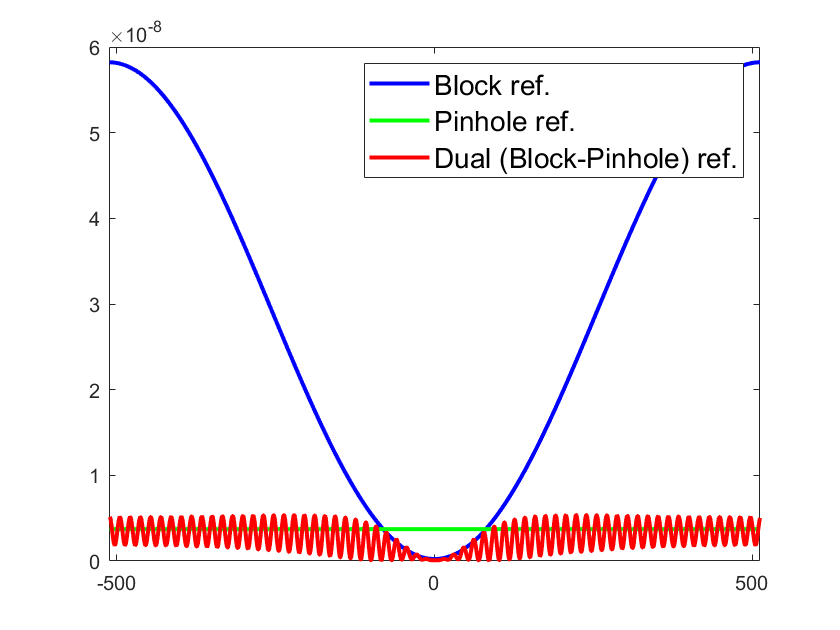}
        \caption{}
          \label{freq-scale-x}
    \end{subfigure}
    \caption{The top row shows colormap plots of the weighting factors $S_R$ for the block, pinhole, and dual references, respectively, when $n=64$ and $m=1024$. The bottom plot shows the three weighting factors along the (four identical) bordering cross-sections of the colormap plots.}
    \label{ref-plot}
\end{figure}
Hence, the influence of the reference scheme on the recovery error is largely determined by $S_R$.

We now derive an convenient analytical expression for $S_{R_{B,P}}$.
\begin{proposition}
For any $k_1,k_2 \in \{0,\dots, m-1\}$, $S_{R_{B,P}}(k_1,k_2)$ is equal to
\begin{align}
 \frac{1}{m^4} \sum_{r,s=0}^{n-1} \abs{\frac{\sigma_r \sigma_s}{\sigma_r^2 \sigma_s^2 + 1} u_r^\top \paren{\mc P_2 F^*}(:, k_1) u_s^\top \paren{\mc P_1 F^*}(:, k_2)   + \frac{1}{\sigma_r^2 \sigma_s^2 + 1} v_r^\top \paren{\mc P_1 F^*}(:, k_1) v_s^\top \paren{\mc P_2 F^*}(:, k_2) }^2.
\end{align}
\end{proposition}
\begin{proof}
From \cref{eq:scale-fact}, $S_{R_{B,P}}(k_1,k_2) = \norm{T_R'(:, mk_1+k_2)}{}^2$ and thus we are interested in the squared column norms of $T_R$. Since $V \otimes V$ is an orthogonal matrix and the Euclidean norm is orthogonally invariant, by \cref{eq:TR_key1}, it is sufficient to consider the squared column norms of
\begin{align*}
    T_R' = \frac{1}{m^2}
    \brac{\paren{\Sigma^2 \otimes \Sigma^2 + I_{n^2}}^{-1} \paren{\Sigma \otimes \Sigma} \paren{U^\top \mc P_2 F^* \otimes U^\top \mc P_1 F^*} +  \paren{\Sigma^2 \otimes \Sigma^2 + I_{n^2}}^{-1} \paren{V^\top \mc P_1 F^* \otimes V^\top \mc P_2 F^*} }.
\end{align*}
For any $\paren{k_1, k_2} \in \set{0, 1, \dots, m-1} \times \set{0, 1, \dots, m-1}$ and the corresponding $k = mk_1 + k_2$,
\begin{align*}
& \norm{T_R'(:, k)}{}^2 \nonumber \\
=\; & \frac{1}{m^4} \sum_{r,s=0}^{n-1} \abs{\frac{\sigma_r \sigma_s}{\sigma_r^2 \sigma_s^2 + 1} \brac{U^\top \mc P_2 F^*}(r, k_1) \brac{U^\top \mc P_1 F^*}(s, k_2)   + \frac{1}{\sigma_r^2 \sigma_s^2 + 1}  \brac{V^\top \mc P_1 F^*}(r, k_1) \brac{V^\top \mc P_2 F^*}(s, k_2) }^2 \\
 =\; & \frac{1}{m^4} \sum_{r,s=0}^{n-1} \abs{\frac{\sigma_r \sigma_s}{\sigma_r^2 \sigma_s^2 + 1} u_r^\top \paren{\mc P_2 F^*}(:, k_1) u_s^\top \paren{\mc P_1 F^*}(:, k_2)   + \frac{1}{\sigma_r^2 \sigma_s^2 + 1} v_r^\top \paren{\mc P_1 F^*}(:, k_1) v_s^\top \paren{\mc P_2 F^*}(:, k_2) }^2,
\end{align*}
as claimed.
\end{proof}

In~\cref{ref-plot}, we compare the $S_R$'s for the single-reference setup (either the pinhole or the block reference), with that of our dual-reference setup. For the single-reference setup $[X, R]$, \cite{HologPROptREF} showed that among three reference choices, the block and pinhole references perform best for low-frequency dominant $Y$ and flat-spectrum $Y$, respectively. Surprisingly, rhe simple idea of including the two references simultaneously and solving the resulting stacked linear system helps to combine the strengths. Indeed, as shown in \cref{ref-plot}, $S_{R_{B,P}}$ approximates the minimum of $S_{R_B}$ and $S_{R_P}$ uniformly over the entire frequency spectrum.

\section{Numerical simulations}  \label{sec:exp}
\begin{figure}[!htbp] \label{test-images}
    \centering
    \begin{subfigure}[b]{0.175\textwidth}
        \includegraphics[width=\textwidth]{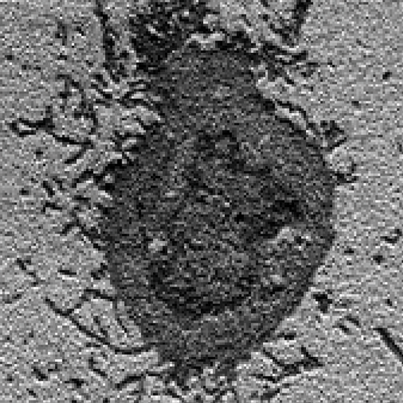}
        \caption{}
        \label{plos2}
    \end{subfigure}
    \begin{subfigure}[b]{0.175\textwidth}
        \includegraphics[width=\textwidth]{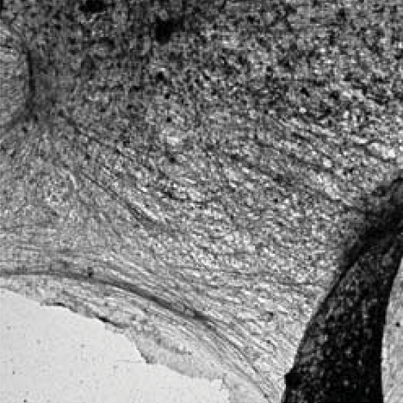}
        \caption{}
          \label{plos1}
    \end{subfigure}
        \begin{subfigure}[b]{0.175\textwidth}
        \includegraphics[width=\textwidth]{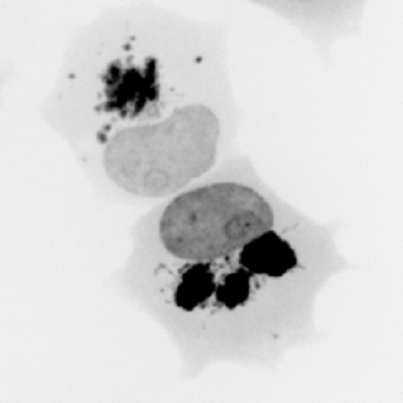}
        \caption{}
          \label{plos3}
    \end{subfigure}
        \begin{subfigure}[b]{0.175\textwidth}
        \includegraphics[width=\textwidth]{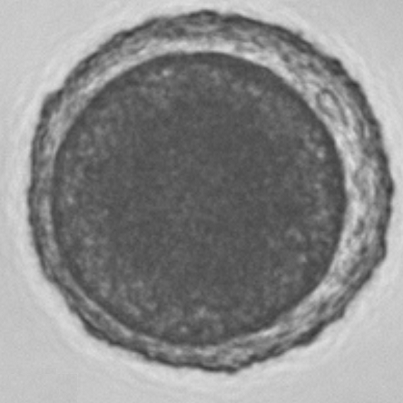}
        \caption{}
          \label{plos4}
    \end{subfigure}
        \begin{subfigure}[b]{0.175\textwidth}
        \includegraphics[width=\textwidth]{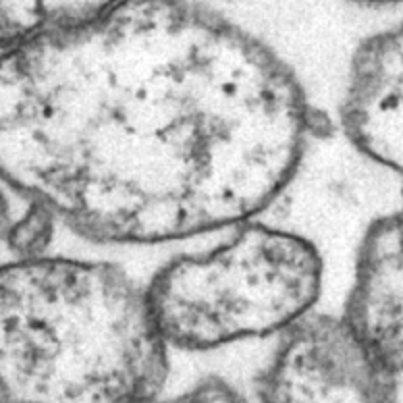}
        \caption{}
          \label{plos5}
    \end{subfigure}
        \begin{subfigure}[b]{0.175\textwidth}
        \includegraphics[width=\textwidth]{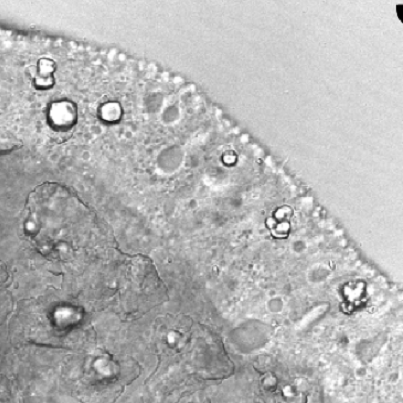}
        \caption{}
          \label{plos6}
    \end{subfigure}
        \begin{subfigure}[b]{0.175\textwidth}
        \includegraphics[width=\textwidth]{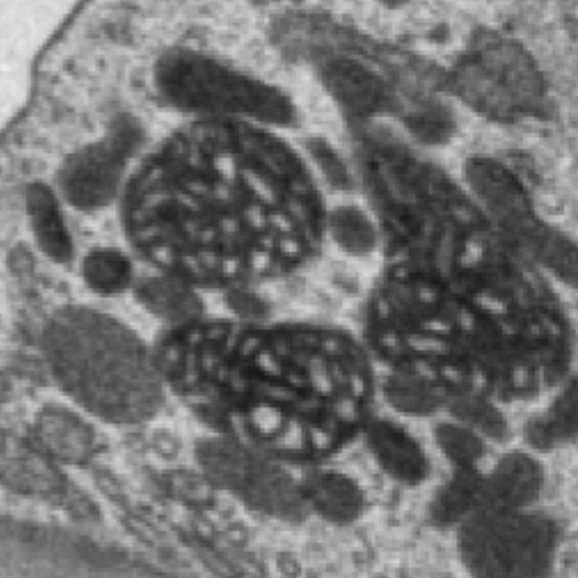}
        \caption{}
          \label{plos7}
    \end{subfigure}
            \begin{subfigure}[b]{0.175\textwidth}
        \includegraphics[width=\textwidth]{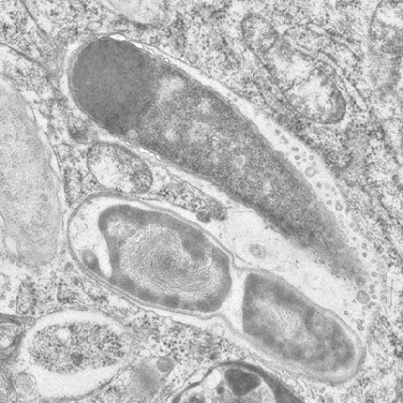}
        \caption{}
          \label{plos8}
    \end{subfigure}
                \begin{subfigure}[b]{0.175\textwidth}
        \includegraphics[width=\textwidth]{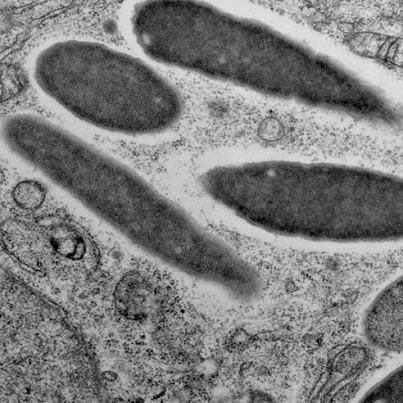}
        \caption{}
          \label{plos9}
    \end{subfigure}
    \caption{CDI specimen images used for numerical simulations in ~\cref{data-table}. Shown are specimens of influenza virus \cite{test-images}, stroma cells \cite{test-images}, mCherry proteins \cite{plos3}, 2-cell embryo \cite{plos4-5}, oocytes \cite{plos4-5}, S. pistillata \cite{plos6}, in-cellular aragonite crystal\cite{plos7}, salmonella WT\cite{plos8-9}, and sifA salmonella strain \cite{plos8-9}, respectively.}
    \label{test-images}
\end{figure}
\begin{table}[!htbp]
\centering
\caption{Empirical and expected values (shown in brackets) for the relative squared errors $\|\wt X-X\|_F^2/\norm{X}{F}^2$ using various recovery methods. The methods are referenced deconvolution with a block, pinhole, and dual-reference, as well as HIO with no reference (HIO(a)), with a block reference (HIO(b)), and with a pinhole reference (HIO(c)). The error values are the reported values rescaled by $10^{-4}$. Test images $X$ of size $64 \times 64$ pixels are the mimivirus shown in \cref{single-ref}, and the images shown in \cref{test-images}. Simulated photon flux data is of size $1024 \times 1024$, with $N_p=1000(1024^2)$ (i.e. 1000 photons per pixel).  }
\begin{tabular}{c|cccccc}
\textbf{Image} & \textbf{Block Ref.} & \textbf{Pinhole Ref.} & \textbf{Dual Ref.} & \textbf{HIO (a)} & \textbf{HIO (b)} & \textbf{HIO (c)}\\
\hline
\textbf{Mimivirus} & 3.70 (3.79) & 46.9 (63.8) & 1.51 (1.45) & 93.7 & 42.8 & 168.1 \\
\textbf{Influenza} & 18.7 (18.5) & 50.7 (31.4) & 4.64 (4.70) & 695.7 & 219.6 & 401.1 \\ \textbf{Stroma cells} & 9.19 (8.91) & 23.1 (44.1) & 2.78 (2.63) & 1607.0 & 110.8 & 2204.8\\
\textbf{mCherry proteins} & 1.84 (1.84) & 139.5 (131.5) & 0.927 (0.908) & 403.1 & 19.1 & 162.4\\
\textbf{Embryo} & 6.30 (6.29) & 54.2 (53.8) & 2.62 (2.71) & 642.4 & 140.9 & 672.5 \\
\textbf{Oocytes} & 7.01 (6.84) & 44.1 (78.3) & 2.66 (2.70) & 883.1 & 96.9 & 895.5\\
\textbf{S. pistallata} & 4.02 (3.93) & 148.8 (83.7) & 1.29 (1.31) & 335.6 & 51.2 & 87.4 \\ \textbf{Aragonite} & 11.6 (11.5) & 52.1 (34.6) & 4.41 (4.41) & 1767.3 & 181.7 & 971.1\\ \textbf{Salmonella WT} & 9.07 (8.81) & 44.3 (60.1) & 3.33 (2.97) & 708.0 & 189.8 & 389.0 \\ \textbf{sifA} & 7.27 (7.15) & 42.6 (54.6) & 2.82 (2.86) & 1765.6 & 221.2 & 678.7\\
\end{tabular}
\label{data-table}
\end{table}

For numerical evaluation of the proposed dual-reference scheme, we use images of $9$ different specimens from diverse sources: influenza virus \cite{test-images}, stroma cells \cite{test-images}, mCherry proteins \cite{plos3}, 2-cell embryo \cite{plos4-5}, oocytes \cite{plos4-5}, S. pistillata \cite{plos6}, in-cellular aragonite crystal\cite{plos7}, salmonella WT\cite{plos8-9}, and sifA salmonella strain \cite{plos8-9}.

To set up, each image is resized to $64 \times 64$, and the pixel values are normalized to $[0, 1]$. The oversampled Fourier transform is taken to be of size $1024 \times 1024$, and the collected noisy data $\widehat{Y}$ obeys the Poisson shot noise model defined in~\cref{eqn:data}. The nominal number of photons is set as $N_p = 1000 \times 1024^2$.

We compare the proposed dual-reference design with the single-reference design studied in~\cite{HologPROptREF}. Specifically, we consider the configuration $[X, R]$ with $R$ being either the pinhole or the block reference. The images are retrieved by the referenced deconvolution algorithm. We also report the performance of the classic hybrid input output (HIO) algorithm on these images, with or without an augmented reference. Performance is measured by the relative recovery error, defined as
\begin{align}
\eps \doteq \frac{\|X - \widehat{X}\|^2}{\norm{X}{}^2}.
\end{align}

The detailed recovery errors for the various schemes are tabulated in \cref{data-table}. It is evident that the dual-reference scheme performs consistently better than other schemes.

\section{Conclusions}
We have proposed a novel dual-reference scheme for holographic CDI, together with a recovery algorithm which provides exact recovery in the noiseless setting. For data corrupted by Poisson shot noise, the dual-reference combines the best features of the block and pinhole references. Numerical experiments on simulated CDI data show the dual-reference scheme provides a smaller recovery error than the leading (single) reference schemes.

\subsection*{Acknowledgments}
The authors are very grateful to Walter Murray, Gordon Wetzstein for many guiding discussions throughout this research.

\appendices

\section{Connection between the continuous and discrete settings}  \label{sec:c2d_app}

The computational imaging literature often only describe the continuous formulation of CDI imaging, without touching on the intrinsic discretization issue due to the finite-grid photon detector and computational feasibility. For the sake of completeness, here we clarify the relevant issues. An exposition similar to \cref{sec:app_sampling} can also be found in Sec. 5.6 of~\cite{GuizarSicairos2010Methods}.

\subsection{Sampling at the detector}  \label{sec:app_sampling}

In CDI setup, suppose the detector is sufficiently far away from the specimen of interest so that the well known \emph{Fraunhofer approximation} applies. Denote the (electromagnetic) field transmitted through the specimen as $I\paren{x, y} \doteq f(x, y, 0)$. Then, the field intensity at the far-field (i.e., $z \gg 0$) detector can be approximated as
\begin{align*}
\abs{f(x,y,z)}^2 \approx \frac{1}{\lambda^2 z^2} |\wh{I}(\frac{x}{\lambda z}, \frac{y}{\lambda z})|^2,
\end{align*}
where $\lambda$ is the wavelength of the electromagnetic radiation and $\wh{I}$ is the continuous 2D Fourier transform of $I$:
\begin{align*}
\wh{I}(u,v)= \int_{\R^2}I(x,y)e^{-i 2\pi (ux+vy)}\; dxdy.
\end{align*}
The detector consists of an array of $\Delta \times \Delta$ square pixel areas, each of which counts the incident photons. Let $p, q$ index the 2D array of pixels. The radiation energy received at the $\paren{p, q}$-th pixel area ($p, q \in \Z$) is given by
\begin{align*}
\frac{\tau \mu}{\lambda^2 z^2}\int_{\paren{x,y} \in [-\Delta/2,\Delta/2]^2}\abs{\wh{I}\paren{\frac{p\Delta+x}{\lambda z},\frac{q\Delta+y}{\lambda z}}}^2\; dxdy,
\end{align*}
where $\tau$ is the collection duration and $\mu$ is the detector quantum efficiency. When $\Delta$ is sufficiently small, the above integral can be well approximated by
\begin{align}  \label{eq:detector_photon_mean}
\frac{\Delta^2 \tau \mu}{\lambda^2 z^2}\abs{\wh{I}\paren{\frac{p\Delta}{\lambda z},\frac{q\Delta}{\lambda z}}}^2,
\end{align}
which can be treated as the mean number of photons measured by the $\paren{p, q}$ pixel location in the collection duration $\tau$.

\subsection{Continuous to discrete Fourier analysis}  \label{sec:app_c2d_fourier}

So measurements of CDI are effectively samples of the spectrum $|\wh{I}\,|^2$ at discrete locations $\paren{p\Delta', q\Delta'}$ for $p, q \in \Z$, where $\Delta' \doteq \Delta/\paren{\lambda z}$. We will write these samples collectively as $Y$ where $Y[p, q] = |\wh I\paren{p\Delta', q\Delta'}|^2$. In practice, phase retrieval is the problem of recovering the complex phases of $Y$, from $Y$ and additional knowledge about $I$.

Recovering these spectrum samples is crucial, as they can be used for recovering $I$ itself. In practical CDI, $I$ is always compactly supported. Without loss of generality, we assume the support is $[-B/2, B/2] \times [-B/2, B/2]$. Since $\wh{I}$ and $I$ are related by 2D continuous Fourier transform, 2D sampling theorem (see, e.g., Chapter 2 of~\cite{Goodman2005Introduction}) implies that whenever
\begin{align}
\frac{\Delta}{\lambda z} \le \frac{1}{B} \Longleftrightarrow \Delta \le \frac{\lambda z}{B},
\end{align}
$I$ can be recovered from samples $\wh I\paren{p\Delta', q\Delta'}$ for all $p, q \in \Z$ via interpolation.

Performing continuous-time Fourier transform on $\wh I\paren{u,v} \sum_{p, q \in \Z}\delta\paren{u - p \Delta', v - q \Delta'}$ reduces to performing discrete-time Fourier transform (DTFT) on $\wh I[p, q]$, where $\delta$ denotes the delta function, either 1D or 2D depending on the context. Moreover, the samples of $\mathrm{DTFT}(\wh I[p, q])$, i.e., samples of $I$, can be computed via discrete Fourier transform. This serves as a high-level justification of our discrete formulation, and is also well established in the digital signal processing literature.

Now we work out the quantitative details. The continuous-time (inverse) Fourier transform on the sampled $\wh I\paren{u,v} \sum_{p, q \in \Z}\delta\paren{u - p \Delta', v - q \Delta'}$ is
\begin{align*}
 J\paren{x, y}
 & = \int_{\R^2} \wh I\paren{u, v} \sum_{p, q \in \Z} \delta\paren{u - p \Delta', v - q \Delta'} e^{i 2\pi \paren{ux + vy}} \; du dv \\
 & = \int_{\R^2} \sum_{p, q \in \Z} \wh I\paren{p \Delta', q \Delta'} \delta\paren{u - p \Delta', v - q \Delta'}e^{i 2\pi \paren{ux + vy}} \; du dv \\
 & = \sum_{p, q \in \Z} \wh I\paren{p \Delta', q \Delta'} \int_{\R^2} \delta\paren{u - p \Delta', v - q \Delta'}e^{i 2\pi \paren{ux + vy}} \; du dv \\
 & = \sum_{p, q \in \Z} \wh I\paren{p \Delta', q \Delta'} e^{i 2\pi \paren{p \Delta' x + q \Delta' y}},
\end{align*}
where at the last step we used the sifting property of the $\delta$ function. We recognize that the last equation represents the (exponential-conjugate) 2D discrete-time Fourier transform on the discrete sequence $\wh I[p, q]$, with the frequency axes rescaled by a factor $2\pi \Delta'$.

Practical detectors have only finite sizes. So we do not have access to the 2D discrete sequence $\wh I[p, q]$, but a truncated version $\wh I_{T}[p, q]$ so that:
\begin{align*}
\wh I_{T}[p, q] =
\begin{cases}
\wh I[p, q]  &   \abs{p} \le D \; \text{and} \;  \abs{q} \le D \\
0     &   \text{otherwise}
\end{cases},
\end{align*}
where we assume the detector consists of $(2D+1)\times (2D+1)$ pixels. Thus, at best we can compute an approximation $\wt J$ to $J$:
\begin{align}
 \wt J \paren{x, y} = \sum_{p, q \in \Z:\; \abs{p} \le D, \abs{q} \le D} \wh I\paren{p \Delta', q \Delta'} e^{i 2\pi \paren{p \Delta' x + q \Delta' y}}.
\end{align}
When the spectrum $\wh I$ decays sufficiently fast\textemdash which is often true in practice\textemdash and $2D+1$ (size of the detector) is relatively large, $\|J - \wh J\|_{L_1}$ tends to be small. Taking the discrete Fourier transform on $\wh I_{T}[p, q]$ (for convenience, we take a less standard convention)
\begin{align}
\sum_{p = -D}^{D} \sum_{q = -D}^D \wh I_{T}[p, q] e^{i2\pi \paren{p n_1 + qn_2}/\paren{2D + 1}}
\end{align}
obviously computes $\wt J$ at locations $\paren{\frac{n_1}{\paren{2D+1} \Delta'}, \frac{n_2}{\paren{2D+1} \Delta'}}$ for all $n_1, n_2 \in \Z$ with $\abs{n_1} \le D, \abs{n_2} \le D$.

To sum up, when the sampling interval on the detector is sufficiently small ($\Delta \le \lambda z/B$) and the detector contains sufficiently large number of pixels (i.e., $D$ large), the discrete formulation for CDI is a reasonable proxy to the continuous formulation. Particularly, recovering the phases of the frequency samples helps to approximate the original field $I$ at discrete sampled points.

\subsection{Continuous vs. discrete signal recovery}

In this subsection, we sketch the correspondence between the discrete signal recovered via the recovery algorithm of this paper versus the continuous specimen signal. Firstly, we show that the discrete autocorrelation obtained is (approximately) a sampled version of the continuous autocorrelation. Secondly, we discuss how the discrete signal $X$ recovered from the discrete autocorrelation approximates the continuous specimen signal $I$.

Let $A_I$ denote the autocorrelation of $I$, i.e.,
\begin{align*}
A_I(s,t)&=\int_{\R^2}I(x,y)\overline{I(x-s,y-t)}\; dxdy.
\end{align*}
Thanks to the convolution theorem of Fourier transform,
\begin{align*}
|\wh{I}(u,v)|^2 &= \int_{\R^2}A_I(x,y)e^{-i 2\pi (ux+vy)}\; dxdy,
\end{align*}
i.e., squared Fourier magnitudes of $I$ are the Fourier transform of $A_I$.

As discussed above, sampling occurs at the detector. On one hand, due to the convolution theorem of Fourier transform,
\begin{align*}
& \int_{\R^2} |\wh I\paren{u, v}|^2 \sum_{p, q \in \Z} \delta\paren{u - p\Delta', v - q \Delta'} e^{i2\pi \paren{ux + vy}} \; du dv \\
=\; &  \int_{\R^2} |\wh I\paren{u, v}|^2 e^{i2\pi \paren{ux + vy}} \; du dv \; \ast \\
& \qquad \; \int_{\R^2}\sum_{p, q \in \Z} \delta\paren{u - p\Delta', v - q \Delta'} e^{i2\pi \paren{ux + vy}} \; du dv  \\
=\; & A_I\paren{x, y} \; \ast \; \frac{1}{\Delta'^2} \sum_{k, \ell \in \Z} \delta\paren{x - \frac{k}{\Delta'}, y - \frac{\ell}{\Delta'}}.
\end{align*}
On the other hand, similar to the argument in \cref{sec:app_c2d_fourier},
\begin{align*}
\int_{\R^2} |\wh I\paren{u, v}|^2 \sum_{p, q \in \Z} \delta\paren{u - p\Delta', v - q \Delta'} e^{i2\pi \paren{ux + vy}} \; du dv
= \sum_{p, q \in \Z} \abs{\wh I\paren{p \Delta', q\Delta'}}^2 e^{i 2\pi \paren{p\Delta' x + q \Delta' y}}.
\end{align*}
Thus,
\begin{align}
A_I\paren{x, y} \; \ast \; \frac{1}{\Delta'^2} \sum_{k, \ell \in \Z} \delta\paren{x - \frac{k}{\Delta'}, y - \frac{\ell}{\Delta'}}
= \sum_{p, q \in \Z} \abs{\wh I\paren{p \Delta', q\Delta'}}^2 e^{i 2\pi \paren{p\Delta' x + q \Delta' y}}.
\end{align}
For sanity check, both sides are spatially periodic with period $\paren{1/\Delta', 1/\Delta'}$. Now since $I$ is supported on $[-B/2, B/2] \times [-B/2, B/2]$, $A_I$ is supported on $[-B, B] \times [-B, B]$. So whenever
\begin{align}
 \frac{1}{\Delta'} \ge 2B \Longleftrightarrow \frac{\Delta}{\lambda z} \le \frac{1}{2B},
\end{align}
there is no aliasing and we can focus on the region $[-B, B] \times [-B, B]$ which contains a scaled version of $A_I$.

To account for the truncation effect, we again assume $|\wh I\paren{p\Delta', q\Delta'}|^2$ gets truncated whenever $\abs{p} > D$ or $\abs{q} > D$ for a certain $D \in \N$. Then, taking the discrete Fourier transform on the truncated sequence, i.e.,
\begin{align}
\sum_{p = -D}^{D} \sum_{q = -D}^D \abs{\wh I\paren{p \Delta', q\Delta'}}^2 e^{i2\pi \paren{pn_1 + qn_2}/\paren{2D +1}}
\end{align}
approximates the values of $A_I$ at locations $\paren{\frac{n_1}{\paren{2D+1}\Delta'}, \frac{n_2}{\paren{2D+1}\Delta'}}$ for all $n_1, n_2 \in \Z$ with $\abs{n_1} \le D$ and $\abs{n_2} \le D$.

Having obtained an (approximate) sampled version of the continuous autocorrelation, the relationship between the corresponding discrete signal $X$ and the continuous specimen signal exactly corresponds to the midpoint rule approximation of an integral (e.g. see \cite{Stoer2002}). Specifically, it was shown in \cite{HologPROptREF} that the relationship between $X$ and $A_{[X,R]}$ that governs signal recovery is given by:
\begin{equation} \label{eq:cross_corr_XR}
A_{[X, R]} \paren{s_1, -n+s_2} = \sum \limits_{t_1=0}^{n-1} \sum \limits_{t_2 = 0}^{n-1} X\paren{t_1, t_2} \ol{R\paren{t_1-s_1, t_2 - s_2}},
\end{equation}
for $s_1, s_2 \in \set{-(n-1), \dots, 0}$ (see Eq. 2.10 of \cite{HologPROptREF} for this exact setup). Now consider the continuous counterpart (with analogous notational conventions), i.e. a continuous specimen and reference composite given by $[X,R] \in [0,B] \times [0,2B]$. The continuous partial autocorrelation can analogously be derived using the procedure in \cite{HologPROptREF}, and is given by:
\begin{equation} \label{eq:cross_corr_XR_cts}
A_{[X, R]} \paren{s_1, -B+s_2} = \int \limits_{t_1=0}^{B} \int \limits_{t_2 = 0}^{B} X\paren{t_1, t_2} \ol{R\paren{t_1-s_1, t_2 - s_2}},
\end{equation}
for $s_1, s_2 \in [-B,0].$
The points of the discrete $X$ in \cref{eq:cross_corr_XR} serve to approximate the continuous expression in \cref{eq:cross_corr_XR_cts} via the midpoint rule approximation for integration (up to a scaling factor). Thus, in the limit that the number of points on $X$ in the summation \cref{eq:cross_corr_XR} approaches a continuum, the discrete and continous counterparts of $X$ exactly coincide.


\bibliographystyle{IEEEtran}
\bibliography{dual-bib}

\begin{thebibliography}{10}
\providecommand{\url}[1]{#1}
\csname url@samestyle\endcsname
\providecommand{\newblock}{\relax}
\providecommand{\bibinfo}[2]{#2}
\providecommand{\BIBentrySTDinterwordspacing}{\spaceskip=0pt\relax}
\providecommand{\BIBentryALTinterwordstretchfactor}{4}
\providecommand{\BIBentryALTinterwordspacing}{\spaceskip=\fontdimen2\font plus
\BIBentryALTinterwordstretchfactor\fontdimen3\font minus
  \fontdimen4\font\relax}
\providecommand{\BIBforeignlanguage}[2]{{%
\expandafter\ifx\csname l@#1\endcsname\relax
\typeout{** WARNING: IEEEtran.bst: No hyphenation pattern has been}%
\typeout{** loaded for the language `#1'. Using the pattern for}%
\typeout{** the default language instead.}%
\else
\language=\csname l@#1\endcsname
\fi
#2}}
\providecommand{\BIBdecl}{\relax}
\BIBdecl

\bibitem{HCDI_SampTA}
D.~A. {Barmherzig}, J.~{Sun}, E.~J. {Cand{\`e}s}, T.~J. {Lane}, and P.-N. {Li},
  ``Dual-reference design for holographic phase retrieval,'' in
  \emph{International Conference on Sampling and Applications}, 2019.

\bibitem{CDI-orig}
J.~Miao, P.~Charalambous, J.~Kirz, and D.~Sayre, ``{Extending the methodology
  of X-ray crystallography to allow imaging of micrometre-sized non-crystalline
  specimens},'' \emph{Nature}, vol. 400, pp. 342--344, jul 1999.

\bibitem{FT-Cambridge}
M.~Saliba, T.~Latychevskaia, J.~Longchamp, and H.~Fink, ``{Fourier Transform
  Holography: A Lensless Non-Destructive Imaging Technique},'' \emph{Microscopy
  and Microanalysis}, vol.~18, no.~S2, pp. 564--565, 2012.

\bibitem{HologPROptREF}
D.~A. {Barmherzig}, J.~{Sun}, E.~J. {Cand{\`e}s}, T.~J. {Lane}, and P.-N. {Li},
  ``{Holographic Phase Retrieval and Optimal Reference Design},'' \emph{arXiv
  e-prints}, p. arXiv:1901.06453, Jan. 2019.

\bibitem{Mimivirus}
\BIBentryALTinterwordspacing
E.~Ghigo, J.~Kartenbeck, P.~Lien, L.~Pelkmans, C.~Capo, J.-L. Mege, and
  D.~Raoult, ``{Ameobal Pathogen Mimivirus Infects Macrophages through
  Phagocytosis},'' \emph{PLOS Pathogens}, vol.~4, no.~6, pp. 1--17, 2008.
  [Online]. Available: \url{https://doi.org/10.1371/journal.ppat.1000087}
\BIBentrySTDinterwordspacing

\bibitem{Oppenheim}
A.~V. Oppenheim and R.~W. Schafer, \emph{Discrete-Time Signal Processing},
  3rd~ed.\hskip 1em plus 0.5em minus 0.4em\relax Upper Saddle River, NJ, USA:
  Prentice Hall Press, 2009.

\bibitem{Strang_CSE}
G.~Strang, \emph{Computational Science and Engineering}.\hskip 1em plus 0.5em
  minus 0.4em\relax Wellesley, MA, USA: Wellesley-Cambridge Press, 2012.

\bibitem{CDI-stats}
\BIBentryALTinterwordspacing
I.~S. Wahyutama, G.~K. Tadesse, A.~T{\"{u}}nnermann, J.~Limpert, and
  J.~Rothhardt, ``{Influence of detector noise in holographic imaging with
  limited photon flux},'' \emph{Opt. Express}, vol.~24, no.~19, pp.
  22\,013--22\,027, sep 2016. [Online]. Available:
  \url{http://www.opticsexpress.org/abstract.cfm?URI=oe-24-19-22013}
\BIBentrySTDinterwordspacing

\bibitem{test-images}
\BIBentryALTinterwordspacing
D.~Kim, T.~J. Deerinck, Y.~M. Sigal, H.~P. Babcock, M.~H. Ellisman, and
  X.~Zhuang, ``Correlative stochastic optical reconstruction microscopy and
  electron microscopy,'' \emph{PLOS ONE}, vol.~10, no.~4, pp. 1--20, 04 2015.
  [Online]. Available: \url{https://doi.org/10.1371/journal.pone.0124581}
\BIBentrySTDinterwordspacing

\bibitem{plos3}
\BIBentryALTinterwordspacing
P.~Paszkowski, R.~S. Noyce, and D.~H. Evans, ``Live-cell imaging of vaccinia
  virus recombination,'' \emph{PLOS Pathogens}, vol.~12, no.~8, pp. 1--28, 08
  2016. [Online]. Available: \url{https://doi.org/10.1371/journal.ppat.1005824}
\BIBentrySTDinterwordspacing

\bibitem{plos4-5}
\BIBentryALTinterwordspacing
M.~Eitel, L.~Guidi, H.~Hadrys, M.~Balsamo, and B.~Schierwater, ``New insights
  into placozoan sexual reproduction and development,'' \emph{PLOS ONE},
  vol.~6, no.~5, pp. 1--9, 05 2011. [Online]. Available:
  \url{https://doi.org/10.1371/journal.pone.0019639}
\BIBentrySTDinterwordspacing

\bibitem{plos6}
\BIBentryALTinterwordspacing
A.~Venn, E.~Tambutte, M.~Holcomb, D.~Allemand, and S.~Tambutte, ``Live tissue
  imaging shows reef corals elevate ph under their calcifying tissue relative
  to seawater,'' \emph{PLOS ONE}, vol.~6, no.~5, pp. 1--9, 05 2011. [Online].
  Available: \url{https://doi.org/10.1371/journal.pone.0020013}
\BIBentrySTDinterwordspacing

\bibitem{plos7}
\BIBentryALTinterwordspacing
T.~D. Mayorova, C.~L. Smith, K.~Hammar, C.~A. Winters, N.~B. Pivovarova, M.~A.
  Aronova, R.~D. Leapman, and T.~S. Reese, ``Cells containing aragonite
  crystals mediate responses to gravity in trichoplax adhaerens (placozoa), an
  animal lacking neurons and synapses,'' \emph{PLOS ONE}, vol.~13, no.~1, pp.
  1--20, 01 2018. [Online]. Available:
  \url{https://doi.org/10.1371/journal.pone.0190905}
\BIBentrySTDinterwordspacing

\bibitem{plos8-9}
\BIBentryALTinterwordspacing
R.~Rajashekar, D.~Liebl, D.~Chikkaballi, V.~Liss, and M.~Hensel, ``Live cell
  imaging reveals novel functions of salmonella enterica spi2-t3ss effector
  proteins in remodeling of the host cell endosomal system,'' \emph{PLOS ONE},
  vol.~9, no.~12, pp. 1--29, 12 2014. [Online]. Available:
  \url{https://doi.org/10.1371/journal.pone.0115423}
\BIBentrySTDinterwordspacing

\bibitem{GuizarSicairos2010Methods}
\BIBentryALTinterwordspacing
M.~Guizar~Sicairos, ``Methods for coherent lensless imaging and x-ray wavefront
  measurements,'' Ph.D. dissertation, University of Rochester. Institute of
  Optics, 2010. [Online]. Available:
  \url{https://search.proquest.com/docview/762414263?accountid=14026}
\BIBentrySTDinterwordspacing

\bibitem{Goodman2005Introduction}
J.~W. Goodman, \emph{Introduction to Fourier optics}.\hskip 1em plus 0.5em
  minus 0.4em\relax Roberts and Company Publishers, 2005.

\bibitem{Stoer2002}
\BIBentryALTinterwordspacing
J.~Stoer and R.~Bulirsch, \emph{Introduction to Numerical Analysis}.\hskip 1em
  plus 0.5em minus 0.4em\relax Springer New York, 2002. [Online]. Available:
  \url{https://doi.org/10.1007/978-0-387-21738-3}
\BIBentrySTDinterwordspacing

\end{thebibliography}

\end{document}